\documentclass[a4paper,10pt]{amsart}
\usepackage[utf8]{inputenc}
\usepackage{amsthm}
\usepackage{amsmath}
\usepackage{bm}
\usepackage{amsfonts}
\usepackage{amssymb}
\usepackage{mathtools}
\usepackage{mathrsfs}
\usepackage{setspace}
\usepackage{xcolor}
\usepackage{textgreek}
\usepackage{thmtools} 

\usepackage[pdfdisplaydoctitle,colorlinks,breaklinks,urlcolor=blue,linkcolor=blue,citecolor=blue]{hyperref} 

\newcommand{\A}{\mathcal{A}}

\newcommand{\C}{\mathbb{C}}
\newcommand{\D}{\mathcal{D}}

\newcommand{\F}{\mathcal{F}}

\renewcommand{\L}{\mathcal{L}}
\newcommand{\M}{\mathcal{M}}
\newcommand{\N}{\mathbb{N}}

\newcommand{\R}{\mathbb{R}}
\renewcommand{\S}{\mathbb{S}}
\newcommand{\T}{\mathbb{T}}
\newcommand{\TT}{\mathcal{T}}
\newcommand{\U}{\mathcal{U}}

\newcommand{\Z}{\mathbb{Z}}

\DeclareMathOperator{\supp}{supp}
\DeclareMathOperator{\diam}{diam}
\let\div\relax
\DeclareMathOperator{\div}{div}

\renewcommand{\epsilon}{\varepsilon}

\newcommand{\obar}[1]{\overline{#1}}
\newcommand{\ubar}[1]{\underline{#1}}

\newcommand{\imm}{\mathrm{i}}

\newcommand{\set}[1]{\left\{#1\right\}}
\newcommand{\pa}[1]{\left(#1\right)}

\newcommand{\abs}[1]{\left|#1\right|}
\newcommand{\norm}[1]{\left\|#1\right\|}
\newcommand{\brak}[1]{\left\langle#1\right\rangle}

\newtheorem{thm}{Theorem}[section]
\newtheorem{definition}[thm]{Definition}

\newtheorem{prop}[thm]{Proposition}

\theoremstyle{remark}
\newtheorem{rmk}[thm]{Remark}

\numberwithin{equation}{section}

\title[Essential Self-Adjointness Euler Point Vortices]{Essential Self-Adjointness of Liouville Operator\\for 2D Euler Point Vortices}
\author[F. Grotto]{Francesco Grotto}
  \address{Scuola Normale Superiore, Piazza dei Cavalieri, 7, 56126 Pisa, Italia}
  \email{\href{mailto:francesco.grotto@sns.it}{francesco.grotto@sns.it}}
\date\today

\begin{document}

\begin{abstract}
 We analyse the 2-dimensional Euler point vortices dynamics in the Koopman-Von Neumann approach. 
 Classical results provide well-posedness of this dynamics involving singular interactions for a finite number of vortices,
 on a full-measure set with respect to the volume measure $dx^N$ on the phase space, which is preserved by the measurable flow
 thanks to the Hamiltonian nature of the system. We identify a core for the generator of the one-parameter
 group of Koopman-Von Neumann unitaries on $L^2(dx^N)$ associated to said flow, the core being made of
 observables smooth outside a suitable set on which singularities can occur.
\end{abstract}

\maketitle

\section{Introduction}

In classical, finite-dimensional Hamiltonian systems whose Hamiltonian function involves singular interaction,
there may be singular trajectories in which, at finite time, the driving vector field diverges.
When this happens only for a negligible set of initial conditions with respect to an invariant measure,
thus a physically relevant measure on phase space, the motion is said to be \emph{almost complete}.
A relevant example is the so called \emph{improbability of collisions} in $N$-body systems, a problem
that has received attention both in classical \cite{ai78,sa71,sa73} and more recent \cite{fk19} works.

In such systems lacking global well-posedness, another natural question is whether the Liouville operator,
that is the time evolution generator for the dynamics of observables, is essentially self-adjoint
on a class of observables smooth in a dense set obtained by removing singular points from the phase space,
\cite[Section X.14]{rs75i}.

The present work concerns the 2-dimensional Euler point vortices system, a Hamiltonian system of first order
differential equations describing the dynamics of point particles (\emph{vortices})
whose interaction potential is singular, perfectly fitting the setting we just outlined.

The point vortices system actually describes the evolution of an incompressible ideal fluid whose vorticity
is concentrated in a finite number of points. We consider the torus $\T^2=\R^2/(2\pi\Z)^2$ as space domain:
if $G=(-\Delta)^{-1}$ is the Green function of the Laplace operator (with zero space average), the vorticity distribution
$\omega_t=\sum_{i=1}^N \xi_i \delta_{x_i(t)}$ with intensities of the vortices $\xi_i\in\R$ and positions $x_i\in\T^2$ satisfying
\begin{equation}\label{pointvortices}
\dot x_i(t)=\sum_{j\neq i}^N \xi_j \nabla^\perp G(x_i(t),x_j(t)),
\end{equation}
defines a weak solution to Euler equations in vorticity form,
\begin{equation}\label{euler}
\begin{cases}
\partial_t \omega + u\cdot \nabla \omega =0\\
\nabla^\perp \cdot u=\omega
\end{cases},
\end{equation}
where we can express the velocity field $u$ in terms of $\omega$ by the Biot-Savart law, 
$u=-\nabla^\perp G\ast \omega$, $\nabla^\perp=(\partial_2,-\partial_1)$.

The point vortices system is a classical model. We refer to the monography \cite{mp94} for most
of the notions we are going to rely on, and to \cite{lionsbook} for an overview of
the statistical mechanics point of view.
Integrable and non integrable behaviours in point vortices systems are also the subject of a considerable literature.
We refer to \cite{ArSt99,ArKi19} for vortices on $\T^2$, \cite{KiNe98} for vortices on $\S^2$
and to \cite{Ar07} for a survey on the topic: complete references can be found in those works,
including a large number of studies on vortices on $\R^2$.

Due to the Hamiltonian nature of the system, the dynamics \eqref{pointvortices} preserves the volume measure $dx^N$ 
on the phase space $(\T^2)^N=\T^{2\times N}$.
Notwithstanding the singularity of the interaction $\nabla^\perp G$, according to the result of D\"urr and Pulvirenti \cite{dp82},
the dynamics is well-posed for initial data in a full measure set of the phase space, 
thus defining a measurable flow $T_t:\T^{2\times N}\rightarrow \T^{2\times N}$ and giving positive answer to
the problem of almost completeness.

Let us consider the one-parameter group of Koopman unitaries $U_t$ associated to such flow,
\begin{equation*}
	U_t f=f\circ T_t, \quad f\in L^2(\T^{2\times N}).
\end{equation*}
In \cite{af03}, the authors defined the Liouville operator associated to the evolution problem \eqref{pointvortices}
on a set of cylinder functions of Fourier modes, and raised the question of essential self-adjointness.
We will discuss the setting of \cite{af03} in comparison to ours in \autoref{sec:configurationspace}.

The main result of the present paper, \autoref{thm:main}, is a proof of essential self-adjointness for the Liouville operator on $L^2(\T^{2\times N})$.
We will consider a domain of smooth functions on full-measure open sets, vanishing in a neighbourhood of singular points of the driving vector field,
on which we are able to explicitly write the generator $A$ of the Koopman group $U_t=e^{\imm t A}$,
and show that such observables form a core for $A$.

Even if we will achieve our aim by means of an approximation noticeably differing from the one of \cite{dp82,mp94},
much of their understanding of the point vortices system will be crucial to our efforts.
Our method also draws ideas from the work \cite{mpp78}, which discusses essential self-adjointness
of Liouville's operator for an infinite particle system with regular interactions. 

The study of the generator of point vortices dynamics might provide some insight in the much more difficult problem
of essential self-adjointness for the generator of Euler dynamics with enstrophy-measure (space white noise) fixed time distributions.
Euler evolution in such low regular regimes has been linked to point vortices dynamics by Central Limit Theorems
\cite{flandoli,Gr19,GrRo19}: only existence of solutions (dating back to \cite{ac90}) is known, and thus essential self-adjointness
of the generator is sought as a first uniqueness result \cite{AlFe02,AlBaFe08,AlBaFeERR}. However, as already remarked in \cite{af03},
the point measures we consider in this paper are singular with respect to white noise. 
We also mention the recent work \cite{GuPe18}, concerning identification of a domain for the singular generator of
stochastic Burgers equation in an infinite dimensional Gaussian space.

For the sake of simplicity, we will first discuss our result in the case of a fixed number $N$ of vortices
on the torus. In \autoref{sec:generalisations} we comment on how one can modify our arguments to cover
different geometries and reference measures, and finally discuss the configuration space approach of \cite{af03} in \autoref{sec:configurationspace}.

\section{The Liouville Operator for Point Vortices Systems}\label{sec:main}

In this section, $\ubar x=(x_1,\dots x_N)\in (\T^2)^N$ are the position of point vortices on $\T^2$,
and $\ubar \xi=(\xi_1,\dots \xi_N)\in\R^N$ their intensities. We denote by $dx^N$ the uniform (Haar's) measure
on $(\T^2)^N$, such that its total volume is normalised to $1$, and by $d(x,y)$ the distance of points $x,y\in\T^2$. 
Also, we will denote $|B|$ the measure of measurable subsets $B\subset (\T^2)^N=\T^{2\times N}$.
All observables are intended as complex valued, and we will denote $L^2(\T^{2\times N})=L^2(\T^{2\times N};\C)$. 
We distinguish the imaginary unit $\imm\in\C$ and the index $i\in \N$ (in italics).
Time $t\in\R$ ranges the whole real line, but for simplicity we will often consider positive times $t>0$,
the other case being completely analogous.

\subsection{Classical Results on Improbability of Collisions}\label{ssec:collisions}

The system \eqref{pointvortices} is an ordinary differential equation in finite dimension, whose vector field is given by
\begin{equation}\label{vectorfield}
	B_i(\ubar x)=\sum_{j\neq i}^N \xi_j K(x_i(t),x_j(t)), \quad \ubar x\in(\T^2)^N,
\end{equation}
where $K=\nabla^\perp G$. The vector field is singular on the generalised diagonal
\begin{equation*}
	\triangle^N=\set{\ubar x\in(\T^2)^N: x_i=x_j \text{ some }i\neq j},
\end{equation*}
because $K(x,y)$ diverges when $x=y\in\T^2$: indeed, we recall that
\begin{equation}\label{greentorus}
    G(x,y)=-\frac1{2\pi}\log d(x,y)+g(x,y),
\end{equation}
with $g\in C^\infty(\T^2\times\T^2)$ a symmetric, zero averaged function. 
Since $G(x,y)=G(x-y,0)$ is translation invariant, the latter representation can be obtained by
solving the equation $-\Delta u(x)=\delta_0(x)$ in a ball $B\subset \T^2$ centred in $0$ with Dirichlet boundary conditions,
and considering the difference $G-u$.

Although $B$ is smooth outside $\triangle^N$, classical well-posedness theorems can only provide existence and uniqueness
of solutions only locally in time. Indeed, if some vortices collapse, that is if a solution reaches $\triangle^N$,
the vector field diverges.
However, $B$ is divergence free (in the sense of distributions) and it thus formally preserves the measure $dx^N$. 
In fact, the point vortices system is Hamiltonian with respect to conjugate coordinates $(x_{i,1},\xi_i x_{i,2})$, 
and Hamiltonian function
\begin{equation*}
H(\ubar x)=\sum_{i< j}^N \xi_i\xi_j G(x_i,x_j)
\end{equation*}
(the interaction energy of vortices). By exploiting this peculiar structure of $B$, it is possible to prove that in fact,
for any fixed --but arbitrary-- choice of intensities $\xi$, the system \eqref{pointvortices} has a global (in time), smooth solution
for almost every initial condition with respect to $dx^N$.

The case in which all intensities $\xi_i$ have the same sign is easier, since the minimum distance between vortices along a trajectory
in phase space can be controlled by means of the Hamiltonian $H$. Indeed, by \eqref{greentorus}, there exist constants $C,C'$
depending on $\ubar \xi,N$ such that
\begin{equation}\label{distancecontrol}
	\min_{i\neq j} |x_i-x_j|\geq C e^{-C' |H(\ubar x)|},
\end{equation}
and since the right-hand side is a first integral of the motion, we can extend local-in-time solutions 
of \eqref{pointvortices} starting from $\ubar x\in \T^{2\times N}\setminus\triangle^N$ to global solutions
which are also smooth in time.

When vortices intensities $\ubar \xi\in\R^N$ take both positive and negative values, there might exist
initial conditions leading to collapse, see \cite[Section 4.2]{mp94} and the references above on integrable motion
of vortices. Indeed, the energy $H(\ubar x)$ gives us
no control whatsoever on the vortices distances along the trajectory of $\ubar x$, since
$H$ include now both positive and negative terms which can cancel out large contributions
of close couples of vortices. 

Almost completeness in the general case of intensities with positive and negative signs is a classical result due to D\"urr and Pulvirenti, \cite{dp82}; 
we also refer to \cite{mp94} for the case of vortices on the whole plane (see \autoref{sec:generalisations} below). 
The following result summarises Theorems 2.1 and 2.2 of \cite{dp82}

\begin{thm}[D\"urr-Pulvirenti]\label{thm:dp}
	Let  $\ubar \xi\in\R^N$ be fixed. 
	There exists a full-measure set $M\subset (\T^2)^N$ and a one-parameter group of 
	maps $T_t:M\rightarrow M$ such that $\ubar x(t)=T_t(\ubar x)$ is the unique,
	smooth solution of \eqref{pointvortices} with initial datum $\ubar x(0)=\ubar x\in M$.
	For all $t\in\R$, $T_t$ defines a measurable, measure preserving, $dx^N$-almost everywhere
	invertible transformation of $(\T^2)^N$.
	
	Define moreover, for $t>0$ and $\ubar x\in (\T^2)^N$,
	\begin{equation*}
	d_t(\ubar x)=\inf_{s\in[0,t]} \min_{i\neq j}|(T_s\ubar x)_i-(T_s\ubar x)_j|.
	\end{equation*}
	Then there exists a constant $C>0$ independent of $c>0$ such that
	\begin{equation}\label{nocollapse}
	\abs{\set{d_t(\ubar x)<c}}\leq \frac{CT}{-\log c}.
	\end{equation}
\end{thm} 

We now briefly review the proof of \autoref{thm:dp}, since its core ideas underlie most of our arguments.
Consider the smooth vector field on $\T^{2\times N}$ given by
\begin{equation}\label{smoothvectorfield}
	B_i^\epsilon(\ubar x)=\sum_{j\neq i}^N \xi_j K_\epsilon(x_i(t),x_j(t)), \quad K_\epsilon =\nabla^\perp G_\epsilon,
\end{equation}
that is the point vortices vector field \eqref{vectorfield} with a smoothed interaction obtained by $G_\epsilon\in C^\infty(\T^2)$
such that:
\begin{equation}\label{gapproxi}
	G_\epsilon|_{B(0,\epsilon)^c}=G|_{B(0,\epsilon)^c}, 
	\quad |\nabla G^\epsilon(x)|\leq |\nabla G(x)|\leq \frac{C}{|x|} \quad \forall x\in\T^2,
\end{equation}
where $B(0,\epsilon)\subset \T^2$ is the ball of radius $\epsilon$ centred in $0$ and $C>0$ is a universal constant. 
We denote by $T^\epsilon_t \ubar x=T^\epsilon(t,\ubar x)$ the flow of the ordinary differential equation
\begin{equation*}
	\begin{cases}
		\dot{\ubar x}(t)=B^\epsilon(\ubar x(t))\\
		\ubar x(0)=\ubar x	   
	\end{cases},
\end{equation*}
which is globally well-posed since its driving vector field is smooth. 
Moreover, we define for $\epsilon>0, t>0$ and $\ubar x\in \T^{2\times N}$,
\begin{equation}\label{depsilon}
d^\epsilon_t(\ubar x)=\inf_{s\in[0,t]} \min_{i\neq j}|(T^\epsilon_s\ubar x)_i-(T^\epsilon_s\ubar x)_j|.
\end{equation}

\begin{proof}[Proof of \autoref{thm:dp}, sketch.]
	Consider the Lyapunov function
	\begin{equation*}
		\L^\epsilon(\ubar x)=\sum_{i\neq j}G_\epsilon(x_i,x_j)=\sum_{i\neq j}G_\epsilon(x_i-x_j),
	\end{equation*}
	which, unlike $H$, controls the minimum distance between vortices:
	\begin{equation}\label{Lcontroldistance}
		d_t^\epsilon(\ubar x)\leq C \exp\pa{\sup_{s\in[0,t]}|\L^\epsilon(T^\epsilon_t\ubar x)|},
	\end{equation}
	with $C>0$ only depending on $N$.
	Fix $\ubar x$ and denote $x_i^\epsilon(t)=(T^\epsilon_t\ubar x)_i$ for the sake of brevity;
	a straightforward computation gives
	\begin{align*}
		\frac{d}{dt}\L^\epsilon(T^\epsilon_t \ubar x)
		&= \sum_{i\neq j} \nabla G^\epsilon(x_i^\epsilon(t)-x_j^\epsilon(t))\\
		&\cdot \pa{\sum_{k\neq i}\xi_k K^\epsilon (x_k^\epsilon(t)-x_i^\epsilon(t))-
		\sum_{\ell\neq j}\xi_\ell K^\epsilon (x_\ell^\epsilon(t)-x_j^\epsilon(t))}\\
	    &=2 \sum \xi_k \nabla G^\epsilon(x_i^\epsilon(t)-x_j^\epsilon(t)) \cdot K^\epsilon (x_k^\epsilon(t)-x_i^\epsilon(t))
	\end{align*}	
	where the sum in the last line is over triples of indices $(i,j,k)$ such that no pair of them coincide.
	This is due to the essential cancellation $\nabla G_\epsilon\cdot K_\epsilon=0$. As a consequence of this and the
	contruction hypothesis \eqref{gapproxi}, $\L^\epsilon$ is uniformly integrable in $\epsilon$, since
	\begin{equation*}
		\norm{\L^\epsilon}_{L^1(\T^{2\times N})} \leq C \int_{\T^{2\times 3}}\frac1{|x-y|}\frac1{|y-z|}dxdydz<C',
	\end{equation*}
	with $C,C'>0$ depending only on $N,\ubar\xi$. Moreover, the smooth flow $T^\epsilon$ preserves $dx^N$
	(its driving vector field is divergence-free because $\div \nabla^\perp G_\epsilon=0$), 
	and thus there exists $C>0$ independent of $\epsilon$ such that
	\begin{equation*}
		\int_{\T^{2\times N}} \sup_{s\in[0,t]}|\L^\epsilon(T^\epsilon_t\ubar x)| dx^N<C.
	\end{equation*}
	This, in combination with Markov inequality and \eqref{Lcontroldistance} produces the crucial estimate,
	for $C>0$ independent of $c>0$,
	\begin{equation}\label{nocollapseepsilon}
	\abs{\set{d_t^\epsilon(\ubar x)<c}}\leq \frac{CT}{-\log c},	
	\end{equation}
	from which \eqref{nocollapse} follows, since $\set{d_t>c}$ is the almost sure limit of $\set{d_t>\epsilon}$.
	
	On the set $\set{d_t^\epsilon>\epsilon}$ the flow of $B^\epsilon$ and $B$ coincide: sending $\epsilon\rightarrow 0$ we obtain the full-measure
	set $\set{d_t>0}$ on which the flow $T_t(\ubar x)$ of $B$ is well-defined, and intersecting the sets ${d_t>0}$ over a sequence
	of times $t_n\uparrow \infty$ (and one $t_n\downarrow-\infty$) we conclude the proof.	
\end{proof}

Let us stress that the almost surely well-defined flow $T_t(\ubar x)$ of \eqref{pointvortices} produced by \autoref{thm:dp},
is such that for all $t,\epsilon>0$, by definition of $G_\epsilon$, 
\begin{equation}\label{coincidence}
	T^\epsilon_s(\ubar x)=T_s(\ubar x) \quad \forall s\in[0,t], \ubar x\in \set{d_t>\epsilon}.
\end{equation}

\subsection{Functional Analytic Setting}\label{ssec:functional}

This paragraph collects abstract definitions and results we are going to apply to point vortices systems.
We assume knowledge of basic notions in the theory of groups of unitary operators on Hilbert spaces,
for which we refer the reader to \cite[Chapter VIII]{rs75i}.

Let $(X,\F,\mu)$ a standard Borel probability space, \emph{i.e.} $X$ is a Polish space and $\F$ the associated Borel $\sigma$-algebra.
The following results establishes a relation between groups of maps on $X$ and groups of operators on $L^2(\mu)=L^2(X,\F,\mu)$.
Its first part, the easier one, is well known as Koopman's Lemma, whereas the second part, a converse implication,
is a relevant result in Ergodic Theory, for the proof of which we refer to \cite{ggm80}.

\begin{thm}\label{thm:koopman}
	Let the mapping
	\begin{equation*}
		\R\times X\ni (t,x)\mapsto T_t(x)\in X
	\end{equation*}
	be such that: for $\mu$-almost every $x\in X$, $t\mapsto T_t(x)$ is a continuous map; for all $t\in\R$, $x\mapsto T_t(x)$ is
	a $\mu$-almost surely invertible, measurable and measure preserving map and for all $t,s\in\R$
	\begin{equation*}
		T_t\circ T_s(x)=T_{t+s}(x)
	\end{equation*}
	(that is, $(T_t)_{t\in\R}$ is a group). Then
	\begin{equation*}
		L^2(X,\F,\mu)\ni f\mapsto U_t f=f\circ T_t \in L^2(X,\F,\mu)
	\end{equation*}
	defines a strongly continuous group of unitary, positive and unit-preserving operators on $L^2(X,\F,\mu)$
	(\emph{Koopman group}).
	
	Conversely, let $(U_t)_{t\in\R}$ be a strongly continuous group of of unitary, positive and unit-preserving operators
	on $L^2(X,\F,\mu)$ with generator $A$; then there exists a group of $\mu$-almost surely invertible, measurable and measure preserving maps 
	$T_t:X\rightarrow X$, $t\in\R$, such that $U_t f=f\circ T_t$ for all $f\in L^2(X,\F,\mu)$;
	moreover, $t\mapsto T_t(x)$ is weakly measurable for all $t\in\R$.
\end{thm}

\begin{rmk}
	It is worth mentioning that the characterisation of Koopman groups is an important problem in Ergodic Theory.
	We refer to \cite{lem12} for a review on the topic, and to the works \cite{gp76,tel17} for a characterisation
	of Koopman groups in terms of properties of their generators.
\end{rmk}

%

Our aim is to identify a core for the generator of vortex dynamics. This problem is intimately linked
to the one of uniquely extending densely defined symmetric operators and essential self-adjointness.
We recall the following terminology.

\begin{prop}\label{prop:uniquenessterms}
	Consider a symmetric linear operator $(L,D)$ on $L^2(\mu)$; each of the following statements implies the next one:
	\begin{itemize}
		\item \emph{Essential self-adjointness}: the closure of $(L,D)$ is self-adjoint;
		\item $L^2(\mu)$ \emph{uniqueness}: there exists a unique
		one-parameter strongly continuous group of unitaries whose generator extends $(L,D)$.
		\item \emph{Markov uniqueness}: there exists a unique
		one-parameter strongly continuous group of unitaries preserving positivity and unit whose generator extends $(L,D)$.
	\end{itemize}	
\end{prop}
While the second implication is trivial, the first one is due to Stone's theorem:
any one-parameter strongly continuous groups of unitaries on a Hilbert space $H$ is generated by a self-adjoint operator.
We also recall that the first two definitions coincide if $(L,D)$ is semi-bounded; however this will not be the case in our discussion.

The basic self-adjointness criterion is the following (see \cite[Theorem VIII.1]{rs75i}).

\begin{prop}\label{prop:esaeasy}
	Let $H$ be a complex Hilbert space, $U_t=e^{\imm t A}$ a strongly continuous unitary group on $H$ and $A$ its generator.
	If $D\subset D(A)$ is a dense linear subset such that $U_t(D)\subseteq D$, then $(A|_D,D)$ is essentially self-adjoint and 
	$D$ is a core for $A$, $\obar{A|_D}=A$.
\end{prop}

We will in fact use a modified version of this criterion: 
the proof is a standard argument, but we report it for the sake of completeness.

\begin{prop}\label{prop:esacriterion}
	Let $H$ be a complex Hilbert space, $U_t=e^{\imm t A}$ a strongly continuous unitary group on $H$ and $A$ its generator.
	If $D\subset D(A)$ is a dense subset, $L=A|_D$ and
	\begin{equation}\label{esacondition}
		\forall t\in\R,\forall f\in D, \quad U_t f\in D\pa{\obar{L}},
	\end{equation}
	then $(L,D)$ is essentially self-adjoint and $D$ is a core for $A$, $\obar L=A$.
\end{prop}

\begin{proof}
	By \cite[Theorem VIII.3]{rs75i}, if $\ker(L^*\pm\imm)=\set{0}$, then $\obar{L}$ is self-adjoint.
	Assume by contradiction that there exists $f\in D(L^*)$ such that $L^*f=\imm f$ (the case of $L^*f=-\imm f$ is analogous). 
	Then, for all $g\in D=D(L)$ it holds
	\begin{equation}
		\frac{d}{dt} \brak{U_t g,f}_H=\brak{\imm A U_t g,f}_H=\brak{\imm \obar L U_t g,f}_H=\brak{U_t g,f}_H,
	\end{equation}
	where the second passage makes use of the hypothesis $U_t g\in D(\obar L)$, and the last one of $L^\ast =(\obar L)^*$.
	The operator $U_t$ is unitary, so the only solution to the above differential equation for $\brak{U_t g,f}$ in $t$ is the constant $0$,
	and thus, since $g$ varies on the dense set $D$, $f=0$.
	
	We are left to show that $\obar L=A$: this follows easily by differentiating in time $e^{\imm t \obar L}$ on $D$
	and noting that the result coincides by definition with the derivative of $U_t$, so that $U_t=e^{\imm t \obar L}$.
\end{proof}

Let us note that condition \eqref{esacondition} can be rephrased as: for all $t\in\R$ and $f\in D$, there exists a sequence $g_n\in D$
such that
\begin{equation}\label{closedcondition}
	g_n\xrightarrow{n\rightarrow\infty} U_t f, \quad L g_n\xrightarrow{n\rightarrow\infty} \obar{L} U_t f
\end{equation}
in the strong topology of $H$. 

\subsection{The Koopman Group for Point Vortices Systems}\label{ssec:koopvortices}

We denote by $T_t$ the group of transformations of $\T^{2\times N}$ defined in \autoref{thm:dp}
--that is the point vortices flow-- and $U_t$ its associated Koopman group for the remainder of this section.
We now define a first set of observables on which we are able to write explicitly the generator of $U_t$,
and which will turn out to be a core for the generator in the simple case where vortices all have positive (or negative) intensity.

\begin{prop}
	The linear subspace
	\begin{equation*}
	\tilde D =\set{f\in C^\infty(\T^{2\times N}): \supp f \cap \triangle^N=\emptyset}.
	\end{equation*}
	is dense in $L^2(\T^{2\times N})$.
	
	Fix $\ubar \xi\in\R^N$. For any $f\in D$ the following expression is well defined as a function in $L^\infty(\T^{2\times N})$:
	\begin{equation}\label{eq:liouvilleop}
	Lf(\ubar x)=-\imm \sum_{i=1}^N \sum_{j\neq i} \nabla_i f(\ubar x)\cdot \xi_j K(x_i-x_j),
	\end{equation}
	where $\nabla_i f$ denotes the gradient in the $i$-th coordinate of $\T^{2\times N}=(\T^2)^N$.
	
	The operator $(L,D)$ is symmetric; moreover,
	if $A$ is the generator of $U_t$, then $D\subset D(A)$ and $L=A|_D$.
\end{prop}

For the sake of clarity, we recall that $\supp f$, the \emph{support} of $f$, is the closure of the set of points on which $f\neq 0$.
Let us also introduce the useful notation
\begin{equation}
	\triangle^N_\epsilon =\set{\ubar x\in\T^{2\times N}: d(x_i,x_j)\leq \epsilon},
\end{equation}
and notice that the support of any $f\in D$ and $\triangle^N_\epsilon$ are disjoint for any small enough $\epsilon>0$.

\begin{proof}
	The density statement is straightforward: smooth functions $C^\infty(\T^{2\times N})$ are dense in $L^2(\T^{2\times N})$, 
	so we only need to show that we can approximate in $L^2$-norm the elements of $C^\infty(\T^{2\times N})$ with the ones of $D$.
	This is readily done by means of Urysohn's lemma 
	---or rather its $C^\infty$ version on smooth manifolds, see \cite[Theorem 3.5.1]{co01})---
	which ensures existence of smooth functions $g_\epsilon$ vanishing on $\triangle^N_\epsilon$ and
	coinciding with a given $g\in C^\infty(\T^{2\times N})$ outside $\triangle^N_{\epsilon'}$ for $0<\epsilon<\epsilon'$.
	
	The expression \eqref{eq:liouvilleop} is well-defined for $f\in D$ since $\nabla f$ vanishes in a neighbourhood of $\triangle^N_\epsilon$,
	and moreover 
	\begin{equation*}
		\norm{Lf}_\infty \leq C_{\xi,N} \norm{f}_{C^1(\T^{2\times N})} \pa{\min_{\ubar x \in \supp f}\min_{i\neq j} |x_i-x_j|}^{-1}<\infty,
	\end{equation*}
	because $K(x,y)\sim |x-y|^{-1}$ for $x\rightarrow y$ in $\T^2$.
	
	As for the symmetry: first one replaces $K=\nabla^\perp G$ with the cut off kernel $K_\epsilon=\nabla^\perp G_\epsilon$
	as in \eqref{smoothvectorfield}. Integration by parts and the fact that $K_\epsilon$ is divergence free readily show that
	\begin{equation*}
		\int_{\T^{2\times N}} \nabla_i f(\ubar x)\cdot K_\epsilon(x_i-x_j) g(\ubar x) dx^N=
		-\int_{\T^{2\times N}} \nabla_i g(\ubar x)\cdot K_\epsilon(x_i-x_j) f(\ubar x) dx^N,
	\end{equation*}
	in which we can send $\epsilon\rightarrow 0$ by bounded convergence. Summing up all contributions, multiplying by $\imm$
	and taking into account complex conjugation in the scalar product of $L^2(\T^{2\times N})$ we conclude that $(L,D)$ is symmetric.
	
	It remains to show the following limit in $L^2(\T^{2\times N})$,
	\begin{equation*}
		\lim_{t\rightarrow 0} \frac{U_t f-f}{t}=Lf, \quad \forall f\in D.
	\end{equation*}
	But thanks to \autoref{thm:dp}, for almost every $\ubar x\in \T^{2\times N}$ we have that $U_t f(\ubar x)=f(T_t \ubar x)$
	is a smooth function of $t$ and
	\begin{equation}
		\left. \frac{d}{dt} f(T_t \ubar x)\right|_{t=0}= \sum_{i=1}^N \nabla_i f(\ubar x) \cdot \dot x_i(0)=Lf(\ubar x),
	\end{equation}
	so that we can conclude by bounded convergence.
\end{proof}

Uniqueness of the flow in the almost-everywhere sense of \autoref{thm:dp} already gives us,
by means of \autoref{thm:koopman}, the following uniqueness result.

\begin{prop}
   For any fixed $\ubar \xi\in\R^N$, $(L,D)$ is Markov unique, and 
   it extends to the self-adjoint generator $A$ of $U_t=e^{\imm tA}$, the Koopman group of $T_t$.
\end{prop}
%

Before we move on, in the next section, to identify a core for the generator of the Koopman group in the general case
$\ubar \xi$, let us analyse the simpler case of vortices with positive (equivalently, negative) intensities, $\ubar \xi \in (\R^+)^N$.
This indeed is a simpler case because the energy $H(\ubar x)$ controls the minimum distance of vortices as noted above in \eqref{distancecontrol}

\begin{thm}\label{thm:positivevortices}
	Let $\ubar \xi \in (\R^+)^N$. Then the operator $(L,D)$ is essentially self-adjoint, 
	and its closure coincides with the generator of $U_t$.
\end{thm}

\begin{proof}
	We apply the classical criterion of \autoref{prop:esaeasy} by showing that $D$ is left invariant by $U_t$,
	that is for any $f\in D$ and $t\in\R$ it holds $f\circ T_t\in D$.
	By \eqref{distancecontrol}, it holds
	\begin{equation*}
		\forall t\in\R, \quad \min_{\ubar x \in \supp f}\min_{i\neq j} \abs{(T_t\ubar x)_i-(T_t\ubar x)_j}
		\geq C\exp\pa{-C' \min_{\ubar x \in \supp f} |H(\ubar x)|}>0,
	\end{equation*}
	with $C,C'>0$ constants depending only on $\ubar \xi$ and $N$, 
	so that for any $\epsilon$ smaller than the right-hand side of the latter inequality, $f\circ T_t=f\circ T^\epsilon_t$,
	with $T^\epsilon$ being the flow of \eqref{smoothvectorfield} as above. This implies that $f\circ T_t$ is still a smooth function
	and that its support is disjoint from $\triangle^N$, which concludes the proof.
\end{proof}

\subsection{A Core for the Liouville Operator: The General Case}\label{ssec:core}
As we mentioned above, when vortices intensities $\ubar \xi\in\R^N$ take both positive and negative values, there might exist
initial conditions leading to collapse. More generally, the minimum distance
of vortices along a globally defined trajectory of the flow might be $0$, that is the configuration
might pass arbitrarily close to $\triangle^N$. 

As a consequence, even if for $f\in D$ the support of $f$ has a positive distance from the diagonal $\triangle^N$,
trajectories starting from $\supp f$ can travel arbitrarily close to $\triangle^N$ in finite time,
and $D$ is thus not invariant for $U_t$. Moreover, there is no clue that $U_t$ should preserve $C^\infty$ regularity.

Instead of \autoref{prop:esaeasy}, we rely in this case on \autoref{prop:esacriterion}, which allows us to check
a sort of ``approximate invariance'' of the candidate core.
The key is in choosing the correct approximation of $U_t$, and
a natural choice might be to consider the Koopman group of the flow $T^\epsilon$ of the smoothed vector field $B^\epsilon$:
unfortunately this choice is inadequate to our purposes, see \autoref{rmk:notdp} below.
 
We now define a new set of observables, which we will prove to be a core for $A$, and a truncated flow
that will serve us to check conditions of \autoref{prop:esacriterion}.

\begin{definition}\label{def:core}
	We denote by $\tilde D$ the linear space of functions $f\in L^\infty(\T^{2\times N})$ such that:
	\begin{itemize}
		\item there exists a version of $f$ and a full-measure open set $M\subset \T^{2\times N}$ on which $f|_M\in C^\infty(M)$,
		and moreover $\nabla f|_M\in L^\infty(M)$;
		\item $f$ vanishes in a neighbourhood of $\triangle^N$.
	\end{itemize}
\end{definition}

\begin{prop}
	The linear subspace $\tilde D$ is dense in $L^\infty(\T^{2\times N})$,
	and for any $\ubar \xi\in\R^N$, $f\in \tilde D$ the following expression is well defined as a function in $L^\infty(\T^{2\times N})$:
	\begin{equation}
	Lf(\ubar x)=-\imm \sum_{i=1}^N \sum_{j\neq i} \nabla_i f(\ubar x)\cdot \xi_j K(x_i-x_j).
	\end{equation}
	Moreover, $(L,D)$ is a symmetric operator and if $A$ is the generator of $U_t$, then $D\subset D(A)$ and $L=A|_D$.
\end{prop}

The proof of the latter Proposition is completely analogous to the one for $D$ above.
The following is the main result of the present paper.

\begin{thm}\label{thm:main}
	Let $\ubar \xi\in [0,\infty)^N$. Then the operator $(L,\tilde D)$ is essentially self-adjoint, 
	and its closure coincides with the generator $A$ of $U_t$.
\end{thm}

Instead of smoothing the driving vector field, we simply stop trajectories of the flow drawing too close to $\triangle^N$.
Since $T^\epsilon:[0,t]\times \T^{2\times N}\rightarrow \T^{2\times N}$ is a smooth function on a compact set,
by definition $d^\epsilon_t:\T^{2\times N}\rightarrow \R$, defined in \eqref{depsilon}, is a continuous function. In particular,
the sets $\set{d^\epsilon_t<c}$, $\set{d^\epsilon_t>c}$ are open subsets of $\T^{2\times N}$.
Moreover, since
\begin{equation*}
	\abs{\bigcup_{c\geq 0}\set{\ubar x: d^\epsilon_t(\ubar x)=c}}=1,
\end{equation*}
the closed sets $\set{d^\epsilon_t=c}$ are negligible for almost all $c\geq 0$.
Let us stress that
\begin{equation*}
	\forall \ubar x\in \set{d^\epsilon_t>\epsilon}=\set{d_t>\epsilon},
	\quad \forall s\in [0,t],
	\quad T^\epsilon_s \ubar x=T_s \ubar x.
\end{equation*}

We define the (lower semicontinuous) function
\begin{equation}\label{deftau}
	\tau_{t,\epsilon}(\ubar x)=
	\begin{cases}
	t &\ubar x\in \set{d_t^\epsilon>\epsilon}=\set{d_t>\epsilon}\\
	0 &\ubar x\in \set{d_t^\epsilon<\epsilon}
	\end{cases},
\end{equation}
and the arrested flow
\begin{equation}\label{arrestedflow}
	R^\epsilon_t\ubar x=T_{\tau_{t,\epsilon}(\ubar x)}\ubar x=
	\begin{cases}
	T_t\ubar x &\ubar x\in \set{d_t^\epsilon>\epsilon}=\set{d_t>\epsilon}\\
	\ubar x &\ubar x\in \set{d_t^\epsilon<\epsilon}
	\end{cases}.
\end{equation}

We can assume without loss of generality that $\abs{\set{d^\epsilon_t=\epsilon}}=0$,
so that \eqref{deftau} and \eqref{arrestedflow} define $\tau_{t,\epsilon},R^\epsilon_t$ on a full-measure open set. Indeed, if $\set{d^\epsilon_t=\epsilon}$
has positive measure, we can redefine $R^\epsilon_t=T_t=T^\epsilon=T^{\epsilon'}$ on $\set{d^\epsilon_t>\epsilon'}$ 
with a slightly larger $\epsilon'>\epsilon$
such that $\set{d^\epsilon_t=\epsilon'}$ is negligible, and the identity outside $\set{d^\epsilon_t>\epsilon'}$:
this does not influence any of the forthcoming arguments.
This being said, we see that $R^\epsilon_t$ has the following properties: for any $\epsilon>0, t\in\R$,
\begin{itemize}
	\item it is a diffeomorphism on the full-measure open set $\set{d^\epsilon\neq \epsilon}$,
	\item it is a discontinuous but measurable transformation of the whole $\T^{2\times N}$,
	\item it is a measure preserving map.
\end{itemize}
Finally, we define the approximating Koopman operators
\begin{equation}
	V^\epsilon_t f(\ubar x)=f(R^\epsilon_t \ubar x)\quad f\in L^2(\T^{2\times N}),
\end{equation}
which are positivity and unit preserving maps taking values in $L^2(\T^{2\times N})$.

\begin{prop}\label{prop:approxi}
	Fix $f\in \tilde D$ and $t\in\R$. Then:
	\begin{itemize}
		\item[(i)] $V^\epsilon_t f \in\tilde D$;
		\item[(ii)] $V^\epsilon_t f$ converges to $U_t f$ in $L^2(\T^{2\times N})$ as $\epsilon\rightarrow 0$;
		\item[(iii)] $A V^\epsilon_t f=L V^\epsilon_t f$ is well-defined since $V^\epsilon_t f\in \tilde D$,
		and it converges to $A U_t f$ in $L^2(\T^{2\times N})$ as $\epsilon\rightarrow 0$.
	\end{itemize}
\end{prop}

\begin{proof}
	Starting from item (i), first of all we notice that $V^\epsilon_t f=f\circ R^\epsilon_t\in L^\infty(\T^{2\times N})$
	because $f\in L^\infty(\T^{2\times N})$. Let $M$ be, as above, the open set on which (a version of) $f$ is smooth, then
	\begin{equation}
		f\circ R^\epsilon_t (\ubar x)=
		\begin{cases}
		f(T_t \ubar x) & \ubar x\in\set{d^\epsilon_t>\epsilon}\cap (T_t^\epsilon)^{-1}M\\
		f(\ubar x) & \ubar x\in \set{d^\epsilon_t<\epsilon}\cap M
		\end{cases}.
	\end{equation}
	The sets on the right-hand side are disjoint since $\set{d^\epsilon_t>\epsilon}\cap \set{d^\epsilon_t<\epsilon}=\emptyset$,
	and open because intersection of open sets. Moreover, since $T^\epsilon_t$ is measure-preserving,
	\begin{equation*}
		\abs{(T_t^\epsilon)^{-1}M}=\abs{M}=1 \quad \Rightarrow \abs{\set{d^\epsilon_t>\epsilon}\cap (T_t^\epsilon)^{-1}M}=
		\abs{\set{d^\epsilon_t>\epsilon}}
	\end{equation*}
	and also $\abs{\set{d^\epsilon_t<\epsilon}\cap M}=\abs{\set{d^\epsilon_t<\epsilon}}$. This shows that $f\circ R^\epsilon_t$
	coincides with a smooth function on a full-measure open set. As for its gradient,
	\begin{equation*}
		\nabla \pa{f\circ R^\epsilon_t}(\ubar x)=
		\begin{cases}
		\nabla f(T^\epsilon_t\ubar x) DT^\epsilon_t(\ubar x) &\ubar x\in\set{d^\epsilon_t>\epsilon}\cap (T_t^\epsilon)^{-1}M\\
		\nabla f(\ubar x) & \ubar x\in \set{d^\epsilon_t<\epsilon}\cap M
		\end{cases},
	\end{equation*}
	where $\norm{DT^\epsilon_t(\ubar x)}_\infty<\infty$, and thus $\nabla\pa{V^\epsilon_t f}\in L^\infty(\T^{2\times N})$
	since $\nabla f\in L^\infty(M)$. 
	By definition, $R^\epsilon_t(\ubar x)=\ubar x$ on $\set{d^\epsilon_t<\epsilon}$, which is a neighbourhood of $\triangle^N$,
	since it contains all $\triangle^N_{\epsilon'}$ for $\epsilon'<\epsilon$; thus on the intersection of $\set{d^\epsilon_t<\epsilon}$
	and the neighbourhood of $\triangle^N$ on which $f$ vanishes, so must vanish also $V^\epsilon_t f$, concluding item (i).
	
	Item (ii) follows directly from \eqref{nocollapseepsilon} and the fact that $U_t$ is unit preserving:
	\begin{align*}
		\norm{U_t f- V^\epsilon_t f}_{L^2}^2&=\int_{\set{d_t^\epsilon<\epsilon}} \abs{U_t f(\ubar x)-f(\ubar x)} dx^N\\
		&\leq 2\norm{f}_\infty^2 \abs{\set{d_t^\epsilon<\epsilon}}
		\leq \frac{C t \norm{f}_\infty^2}{-\log \epsilon}.
	\end{align*}
	
	Let us now consider how the generator $A$ acts on $V^\epsilon_t f$. By definition,
	\begin{equation*}
		AV^\epsilon_t f(\ubar x)= \left.\frac{d}{ds}\right|_{s=0} U_s V^\epsilon_t f(\ubar x)
		=\left.\frac{d}{ds}\right|_{s=0} f(R^\epsilon_t T_s \ubar x).
	\end{equation*}
	For a fixed $\ubar x$ in the open set $\set{d^\epsilon_t>\epsilon}$, $T_s\ubar x$ is well-defined for $s$ in a neighbourhood of $0$, 
	and it is a smooth function in such time interval. 
	Thus, for small enough $s$ depending on the $\ubar x$ we are fixing, $T_s\ubar x\in \set{d^\epsilon_t>\epsilon}$,
	and the same is true if $\ubar x \in \set{d^\epsilon_t\epsilon}\setminus\triangle^N$ (we are removing the closed negligible diagonal $\triangle^N$).
	As a consequence, for all $\ubar x \in\set{d^\epsilon_t>\epsilon}$,
	\begin{equation*}
		\left.\frac{d}{ds}\right|_{s=0} f(R^\epsilon_t T_s \ubar x)=\left.\frac{d}{ds}\right|_{s=0} f(T_t T_s \ubar x)
		=\left.\frac{d}{ds}\right|_{s=0} f(T_s T_t \ubar x)=Lf(T_t\ubar x)=U_tLf(\ubar x),
	\end{equation*}
	and analogously for $\ubar x \in \set{d^\epsilon_t\epsilon}\setminus\triangle^N$,
	\begin{equation*}
		\left.\frac{d}{ds}\right|_{s=0} f(R^\epsilon_t T_s \ubar x)=\left.\frac{d}{ds}\right|_{s=0} f(T_s \ubar x)
		=Lf(\ubar x).
	\end{equation*}
	We thus see that, for $\ubar x$ in a full-measure set,
	\begin{equation*}
		LV^\epsilon_t f(\ubar x)=V^\epsilon_t Lf.
	\end{equation*}
	Since a strongly continuous unitary group always commutes with its generator on the domain of the latter,
	and since $Af=Lf$ for $f\in \tilde D$,
	\begin{align*}
		\norm{A U_t f- A V^\epsilon_t f}_{L^2}^2&=\norm{U_t L f- V^\epsilon_t L f}_{L^2}^2
		=\int_{\set{d_t^\epsilon<\epsilon}} \abs{U_t L f(\ubar x)-L f(\ubar x)} dx^N\\
		&\leq 2\norm{L f}_\infty^2 \abs{\set{d_t^\epsilon<\epsilon}}
		\leq \frac{C t \norm{f}_{C^1}^2}{-\log \epsilon},
	\end{align*}
	where $C$ is a constant depending on $N$ and $\ubar \xi$. This concludes the proof of (iii).	
\end{proof}

\autoref{thm:main} is a direct consequence of \autoref{prop:esacriterion} and \autoref{prop:approxi}.
Indeed, for fixed $f\in \tilde D$ and $t\in\R$, the we have shown that $V^\epsilon_t f$ satisfies condition \eqref{closedcondition},
and thus $\tilde D$ is a core for $A$.

\subsection{Considerations on unsuccessful approaches}\label{rmk:notdp}

In the proof of \autoref{thm:main} we use in an essential way the peculiar structure of our approximating flow $R^\epsilon_t$
in items (i) and (iii), while (ii) still holds true if we replace $U_tf$ with $U^\epsilon_t f=f\circ T^\epsilon_t$,
the approximating flow of \cite{dp82}, for any smooth $f\in D$. There are two reasons why we are not able to
treat the latter setting. 

Using the fact that $T^\epsilon_t$ is smooth one can show with some care that $U^\epsilon_t$ preserves $D$.
This and estimate \eqref{nocollapseepsilon} would show that $D$ is a core for $A$ provided that we can also show that $AU^\epsilon_t f$
strongly converges to $AU_tf$ for fixed $f\in D,t\in\R$ (\emph{cf.} \autoref{prop:approxi}).
Since $U_tf$ and $U^\epsilon_tf$ coincide on $\set{d^\epsilon_t>\epsilon}$, we only need to evaluate their
difference on $\set{d^\epsilon_t\leq\epsilon}$. The set over which we integrate has small measure $t\log\pa{\frac1\epsilon}$,
but if we try to bound $LU^\epsilon_tf$ uniformly in $\ubar x$ ($LU_tf=U_t Lf$ is uniformly bounded since $Lf$ is),
we are led to control terms including $\norm{DT^\epsilon_t}_\infty$: since the vector field $\norm{B^\epsilon}_\infty\simeq \epsilon^{-2}$,
we get $\norm{DT^\epsilon_t}_\infty\simeq e^{C \epsilon^{-2}}$, which is way too large to be compared with the measure of the integration set.
Considering estimates in $L^2$ or $L^p$ norms does not seem to solve the issue, either.

We have seen above how an abrupt truncation of the flow allows us to show that $\tilde D$ is a core for $A$,
and it is clear that allowing functions of lower regularity was necessary to employ this kind of approximation.
We further mention only one more smooth approach. One might define the vector field
\begin{equation}
	B^\delta(\ubar x)= M^\delta(\ubar x)B(\ubar x),
\end{equation} 
with $M^\delta\in C^\infty(\T^{2\times N})$ vanishing on a $\delta$-neighbourhood of $\triangle^N$ and taking value $1$
far from it. The Koopman operators $V^\delta_t$ of its associated flow would preserve $D$ and strongly converge to $U_t$;
however, $L$ and $V^\delta_t$ would not commute unless $M^\delta$ is a first integral of the motion.
As there can not be invariants of the vortex motions controlling the minimum distance of vortices (as $M^\delta$ would do)
in the case of coexisting positive and negative vortices, we would not be able to continue the proof as we did
for \autoref{thm:main}, and thus have to face explicit computations, in which the difficulties of the same kind of 
the ones outlined above appear.

\section{Generalisations}\label{sec:generalisations}

\subsection{Point Vortices on the Sphere}\label{ssec:vorticessphere}
All the arguments above still work with almost no modifications when the torus $\T^2$ is replaced with a smooth compact surface
with no boundary, such as the sphere $\S^2=\set{x\in\R^3:|x|=1}$ (to be regarded as an embedded surface). On $\S^2$, $d\sigma$ is the Riemannian volume,
so that $\int_{\S^2}d\sigma=1$, and $x\cdot y,x\times y$ respectively denote scalar and vector products in $\R^3$.
Euler equations on $\S^2$ are given by, for $x\in\S^2$,
\begin{equation*}
\begin{cases}
\partial_t \omega(x,t)=x\cdot \pa{\nabla\psi(x,t)\times\nabla\omega(x,t)},\\
-\Delta \psi(x,t)=\omega(x,t).
\end{cases}
\end{equation*}
Here $\Delta$ is the Laplace-Beltrami operator, and we have to supplement the Poisson equation for the \emph{stream function} $\psi$
with the zero average condition. The Green function of $-\Delta$ is simply given by
\begin{equation*}
-\Delta G(x,y)=\delta_y(x)-1, \quad G(x,y)=-\frac1{2\pi} \log |x-y|+c,
\end{equation*}
$c\in\R$ a universal constant making $G$ zero-averaged. To satisfy in weak sense Euler equations, the point vortices vorticity distribution 
$\omega=\sum_1^N \xi_i \delta_{x_i}$ must evolve according to
\begin{equation}\label{vorticessphere}
\dot x_i=\frac1{2\pi}\sum_{i\neq j}^N\xi_j \frac{x_j\times x_i}{|x_i-x_j|^2},
\end{equation}
which is still a Hamiltonian system with
\begin{equation*}
H(x_1,\dots x_N)=\sum_{i< j}^N \xi_i\xi_j G(x_i,x_j).
\end{equation*}
In fact, setting $K(x,y)=\frac1{2\pi}\frac{x\times y}{|x-y|^2}$, \eqref{vorticessphere} takes the same form of \eqref{pointvortices},
and $K$ is still a skew-symmetric, divergence free function on $\S^2$ (divergence being the adjoint of the gradient operator on functions of $\S^2$).
We refer to \cite{PoDr93} for a more complete discussion of this setting.
It is easy to see that all the features we relied on in \autoref{sec:main} are still present:
\begin{itemize}
	\item[(i)] the flow is locally well posed when positions of vortices do not coincide, and it is measure-preserving because of the
	Hamiltonian nature of the equations;
	\item[(ii)] the crucial cancellation leading to integrability of the Lyapunov function $\L(x_1,\dots x_N)=\sum_{i\neq j} G(x_i,x_j)$
	and thus required for the proof of \autoref{thm:dp} to work still takes place, in this case because $(x\times y)\perp (x-y)$
	for any $x,y\in\S^2$;
	\item[(iii)] as a consequence, the almost-surely well defined point vortices flow $T_t$ coincide with a smooth one on open sets of large measures,
	so that we can implement again the strategy of \autoref{ssec:core}.
\end{itemize}

\subsection{Point Vortices on Bounded Domains}
Let $\D\subset \R^2$ be a bounded domain with smooth boundary and Lebesgue measure $|\D|=1$, $G(x,y)$ the Green
function of $-\Delta$ on $\D$ with Dirichlet boundary conditions, which can be represented as the sum of its free version
$G_{\R^2}(x,y)=-\frac{1}{2\pi}\log|x-y|$ and the harmonic extension in $\D$ of the values of $G_{\R^2}$ on $\partial \D$,
\begin{equation}\label{greendomain}
G(x,y)=-\frac{1}{2\pi}\log|x-y|+g(x,y), \quad 
\begin{cases}
\Delta g(x,y)=0 & x\in \D\\
g(x,y)=\frac{1}{2\pi}\log|x-y| &x\in \partial D	
\end{cases}
\end{equation}
for all $y\in \D$. Both $G$ and $g$ are symmetric, and maximum principle implies that
\begin{equation}\label{harmoniccontbound}
-\frac{1}{2\pi}\log(d(x)\vee d(y))\leq g(x,y)\leq \frac{1}{2\pi} \log\diam(\D),
\end{equation}
with $d(x)$ the distance of $x\in \D$ from the boundary $\partial \D$.

The motion of a system of $N$ vortices with intensities $\xi_1,\dots,\xi_N\in\R$ and positions $x_1,\dots,x_N\in \D$
is governed by the Hamiltonian function
\begin{equation*}
H(x_1,\dots,x_n)=\sum_{i< j}^N \xi_i\xi_j G(x_i,x_j)+\sum_{i=1}^{N}\xi_i^2 g(x_i,x_i),
\end{equation*}
leading to the system of equations
\begin{equation*}
	\dot x_i(t)=\sum_{j\neq i}^N \xi_j \nabla^\perp G(x_i(t),x_j(t))+\xi_i^2 \nabla^\perp g(x_i,x_i).
\end{equation*}
The additional (with respect to the cases with no boundary) self-interaction terms involving $g$ are due to the presence of
an impermeable boundary: it is thanks to these terms that the system satisfies (in weak sense) Euler's equations. 
We refer to \cite[Section 4.1]{mp94} for a thorough motivation of this fact.

In this setting, the relevant features (i)--(iii) we individuated in the last paragraph are still present,
but the boundary enters as an additional singular set of the vector field, and thus our arguments
must take it into account. Without going into details, we just mention the relevant required modifications:
\begin{itemize}
	\item the smoothed vector field $B^\epsilon$ and its associated flow $T^\epsilon$ of \autoref{ssec:collisions}
	must be defined by smoothing both $\log|\cdot|$ and $g$ in \eqref{greendomain}: $B^\epsilon$ will coincide with the 
	original vortices vector field $B$ whenever vortices are at least $\epsilon>0$ apart from each other and from the boundary;
	\item functions of $D$ must now satisfy $\supp f\subset \D^N\setminus \triangle^N$, where $\triangle^N$ is the
	diagonal set of $\D^N$ (the definition being the same as in the torus case), whereas function of 
	$\tilde D$ must vanish not only around $\triangle^N$, but also in a neighbourhood of $\set{\ubar x\in \obar{\D}^N: \exists i: x_i\in \partial \D}$.
\end{itemize}

\subsection{Gibbsian Ensembles and Vortices on the Whole Plane}
We now return to point vortices on $\T^2$. Besides the uniform measure $dx^N$ on $\T^{2\times N}$,
the point vortices flow $T_t$ also preserves all \emph{(Canonical) Gibbs measures} defined by
\begin{equation}
	 d\mu_{\beta,N}(\ubar x)= \frac1{Z_{\beta,N}}e^{-\beta H(\ubar x)}dx^N, \quad 
	 Z_{\beta,N}=\int_{\T^{2\times N}}e^{-\beta H(\ubar x)}dx^N.
\end{equation}
In fact, the above density is integrable as soon as $|\beta|\leq C(N,\ubar \xi)$, $C>0$ being a constant depending only on $N$
and the intensities, and the \emph{partition function} $Z_{\beta,N}>0$ is there to make $\mu_{\beta,N}$ a probability measure,
see \cite{GrRo19,lionsbook}.

When $\beta>0$, $\mu_{\beta,N}$ gives more weight to configurations where vortices of the same sign are far from each other,
but positive and negative vortices are close. Vice-versa, if $\beta<0$, vortices of the same sign tend to cluster.
Invariance of $\mu_{\beta,N}$ is an easy consequence of the one of $H(\ubar x)$ and $dx^N$, and can be achieved by
considering the smoothed vortices interaction $B^\epsilon$ with Hamiltonian $H^\epsilon$ and sending $\epsilon\rightarrow 0$.

Whatever $\beta$ is, since $\mu_{\beta,N}$ is absolutely continuous with respect to $dx^N$, the flow $T_t$ is still
globally well-defined on a full-measure set. However, the density of $\mu_{\beta,N}$ is singular in $\triangle^N$
(save for trivial cases), 
so uniform integrability of the Lyapunov functions $\L^\epsilon$ in \autoref{thm:dp} is spoiled.
As a consequence, the arguments in \autoref{ssec:core} also fail.

Let us now spend a few words on point vortices on $\R^2$. The system is given by \eqref{pointvortices} with
$G(x)=-\frac1{2\pi}\log|x|$, and it is well-posed for almost all initial conditions with respect to the product
Lebesgue measure provided that no subset of the intensities $\set{\xi_1,\dots \xi_N}$ sums to zero, see \cite{mp94}.
The latter condition ensures that vortices can not travel to infinity in finite time.

The product Lebesgue measure on $\R^{2\times N}$ is not a probability measure, so we are led to look for an integrable density on $\R^2$
left invariant by the dynamics. To the best of our knowledge, this is only achieved by the Gaussian measure
\begin{align*}
	d\mu_{\alpha,\eta,N}(\ubar x)&=\frac1{Z_{\alpha,\eta,N}}e^{-\eta\cdot M(\ubar x)-\alpha I(\ubar x)}dx^N, 
	\quad \eta\in\R^2, \alpha\in\R^+,\\
	M(\ubar x)&=\sum_{i=1}^N \xi_i x_i, \qquad I(\ubar x)=\sum_{i=1}^N \xi_i |x_i|^2,
\end{align*}
when \emph{all vortices are positive}, $I$ and $M$ being first integrals of vortices motion, the \emph{moment of inertia}
and \emph{centre of vorticity} (see \cite[Section 5.3]{lionsbook}). The interaction energy $H$ can be
also added to the Gibbs exponential, but this is not a substantial modification.
As we have seen above, the case of positive vortices can be dealt with by exploiting conservation of energy,
so we shall not discuss it further. Unfortunately, the more interesting case of arbitrary signs
seems to be impossible to include in our discussion.

\subsection{The Configuration Space and Non-Uniqueness}\label{sec:configurationspace}

In the point vortices time evolution, the number and intensities of vortices are constant --- at least when no vortices collide, as we will see.
As a consequence, everything we said still applies if instead of fixing $N,\ubar \xi$ we choose them at random,
provided that all objects are well defined. In order to discuss an arbitrary number of vortices,
one can consider the phase space
\begin{equation*}
 \bigcup_{N\geq 0} (\T^2\times \R)^N,
\end{equation*}
on which, conditioned to the random choice of $N$, to be made for instance with a sample of a Poisson distribution,
we consider the product measures $dx^N\otimes \nu^{\otimes N}$, with $\nu$ the probability law of a single intensity $\xi_i\in\R$.

An equivalent (up to symmetrisation of products) point of view is the \emph{configuration space} setting, 
in which one looks at the law of the vorticity distribution
$\omega=\sum_{i=1}^N \xi_i \delta_{x_i}$ (the empirical measure of vortices) under the law of the aforementioned ensemble of vortices.
This is the approach of \cite{af03}.
Let us define
\begin{equation*}
\Gamma=\bigcup_{N\geq 0}\Gamma_N, 
\quad \Gamma_N=\set{\gamma=\sum_{i=1}^N \xi_i \delta_{x_i}:\xi_i\in\R, x_i\in\T^2, x_i\neq x_j\text{ if }i\neq j},
\end{equation*}
to be regarded as a subset of finite signed measures $\M(\T^2)$.
There is a one-to-one correspondence between elements of $\Gamma_N$ and classes of equivalence of $(\T^2\times \R)^N$
up to permutations. 
Let $\nu$ be a probability measure on $\R$ with finite second moment and $\lambda>0$;
we define the measure $\mu_N$ on $\Gamma_N$ as the image of $dx^N\otimes \nu^{\otimes N}$ on $(\T^2\times \R)^N$ under the aforementioned correspondence,
and then we define $\mu$ on $\Gamma$ as
\begin{equation*}
	\mu=e^{-\lambda}\sum_{N\geq 0} \frac{\lambda^N}{N!}\mu_N.
\end{equation*}
Equivalently, $\mu$ can be realised by considering a Poisson point process on $\T^2\times \R$ with intensity measure $\lambda dx\otimes d\nu$,
the samples of which are vectors $\pa{x_1,\xi_1,\dots x_N,\xi_N}$, and setting $\mu$ to be the image law under 
the map $\gamma=\sum_{i=1}^N \xi_i \delta_{x_i}$. We refer to \cite{AlKoRo98} for a complete discussion of Poisson processes and the configuration space.

By \autoref{thm:dp}, for $\mu$-almost every $\gamma\in \sum_{i=1}^N \xi_i \delta_{x_i}\Gamma$ the point vortices flow
with initial positions $x_i$ and intensities $\xi_i$ is globally well-defined. Moreover, the flow defines a group of invertible
measurable maps $\TT_t:\Gamma\rightarrow\Gamma$, the cursive to distinguish it from the flow $T_t$ on $L^2(\T^{2\times N})$ in \autoref{sec:main}.
The map $\TT_t$ preserves $\mu$ since it leaves each $\Gamma_N$ invariant, 
and for fixed $N$ the point vortices evolution does not change intensities and preserves the
product measure on the torus.

The main contribution of \cite{af03} is an explicit expression of the generator of the Koopman group $U_t$ on $L^2(\Gamma,\mu)$ associated to $\TT_t$,
on the set of cylinder functions on Fourier modes. In order to comment the problem of essential self-adjointness in this setting
we now repeat their result: we do so perhaps in a more concise way, by means of an important symmetrisation first introduced in the works of
Delort and Schochet \cite{del91,sch95} to give meaning to weak solutions of Euler equations in low regularity regimes.

Indeed, the weak formulation for Euler equations in vorticity form \autoref{euler} against a smooth test $\phi\in C^\infty(\T^2)$ is given by
\begin{align*}
	\brak{\phi,\omega_t}-\brak{\phi,\omega_0}
	&=\int_0^t \int_{\T^{2\times 2}} K(x-y)\omega_s(y) \omega_s(x) \nabla\phi(x) dxdyds\\
	&=\int_0^t \brak{(K\ast \omega_s)\omega_s,\nabla\phi}ds,
\end{align*}
where $\brak{\cdot,\cdot}$ denotes the inner product of $L^2(\T^2)$. For smooth solutions of Euler equations, 
one can symmetrise the variables $x,y$ ---keeping in mind that $K$ is skew-symmetric--- to obtain the expression
\begin{align}\label{dssymm}
	\brak{\phi,\omega_t}-\brak{\phi,\omega_0}&=\int_0^t \int_{\T^{2\times 2}} H_\phi(x,y)\omega_s(x)\omega_s(y)dxdyds\\ \nonumber
	&=\int_0^t \brak{H_\phi,\omega_s\otimes\omega_s}ds,\\
	H_\phi(x,y)&=\frac12 (\nabla\phi(x)-\nabla\phi(y))\cdot K(x-y), \quad x,y\in\T^2,
\end{align}
where $H_\phi(x,y)$ is a symmetric function with zero average in both variables and smooth outside the diagonal set $\triangle^2$,
where it has a jump discontinuity. Because of this, by interpreting brackets $\brak{\cdot,\cdot}$ as suitable duality couplings,
one can give meaning to Euler equations when vorticity $\omega$ has low space regularity.

One such example is the point vortices system: as detailed in \cite{scho96}, 
the empirical measure $\omega=\sum_{i=1}^N \xi_i \delta_{x_i}$
with $x_i$ evolving as in \eqref{pointvortices} satisfies \eqref{dssymm} if we assume that $H_\phi(x,x)=0$,
thas is if we neglect self-interactions of vortices. 
More precisely, brackets $\brak{H_\phi,\cdot}$ are to be interpreted as duality couplings between continuous functions and
measures on $\T^{2\times 2}\setminus\triangle^2$.
We will discuss this further below, here we also recall that \eqref{dssymm}
was also used as an alternative to the Fourier approach of \cite{ac90} to give meaning to white noise and Poissonian solutions in \cite{flandoli,Gr19}.

We define \emph{local observables} on $\Gamma$ as the family $\F$ of functions of the form
\begin{equation}\label{local}
	F(\gamma)=f(\brak{\phi_1,\gamma},\dots \brak{\phi_n,\gamma}),
\end{equation}
where $f\in C_c^\infty(\C^n,\C)$ and $\phi_1,\dots \phi_n\in C^\infty(\T^2,\C)$, the brackets $\brak{\cdot,\cdot}$ denoting
coupling of continuous functions and measures. In \cite{af03} the functions $\phi_k$ were chosen in the Fourier orthonormal basis,
but this would not change anything in our discussion.

\begin{prop}
	Let $U_t$ be the Koopman group on $L^2(\mu)$ associated with $\TT_t$, and $\A$ be its generator.
	For any $F\in\F$ of the form \eqref{local}, the following expression defines an observable in $L^2(\mu)$,
	\begin{equation}
		\L F(\gamma)=-\imm \sum_{k=1}^n \partial_k f(\brak{\phi_1,\gamma},\dots \brak{\phi_n,\gamma})
		\brak{H_{\phi_k},\gamma\otimes\gamma}.
	\end{equation}
	The operator $(\L,\F)$ is symmetric, $\F\subseteq D(\A)$, and $\A|_\F=\L$. Moreover, $(\L,\F)$ is Markov unique,
	that is $\A$ is the unique self-adjoint extension generating a strongly continuous, positivity and unit preserving group
	of unitaries, which is $\U=e^{\imm t\A}$.
\end{prop}

\begin{proof}
	To show that $LF\in L^2(\mu)$, since $\partial_k f$ is uniformly bounded, we just need to compute for $\phi\in C^\infty(\T^2,\C)$,
	\begin{align*}
		&\int \brak{H_{\phi},\gamma\otimes\gamma}^2d\mu_N(\gamma)
		=\int_{\T^{2\times N}}\int_{\R^N} \pa{\sum_{i\neq j} \xi_i\xi_jH_\phi(x_i,x_j)}^2 dx^N d\nu^N(\ubar \xi)\\
		&\qquad =\int_{\T^{2\times N}}\int_{\R^N}\sum_{i\neq j, \ell\neq k} 
		\xi_i\xi_j\xi_\ell\xi_k H_\phi(x_i,x_j)H_\phi(x_\ell,x_k)dx^N d\nu^N(\ubar \xi)\\
		&\qquad =2\sum_{i\neq j} \int_{\R^2}\xi_i^2\xi_j^2d\nu(\xi_i)d\nu(\xi_j) \int_{\T^{2\times 2}} H_\phi(x,y)^2dxdy
		\leq C_{\phi,\nu} N^2,
	\end{align*}
	where we made essential use of the fact that $H_\phi$ is zero-averaged in both variables, so the only non vanishing terms
	in the double sum are the ones with $i=\ell, j=k$ (or vice-versa). We also recall that $\xi_i$ are independent with finite
	second moments. We mention that this was a crucial computation in the works \cite{flandoli,Gr19,GrRo19}. From here,
	\begin{equation*}
		\int \brak{H_{\phi},\gamma\otimes\gamma}^2 d\mu(\gamma)\leq e^{-\lambda}\sum_{N\geq 0} \frac{\lambda^N}{N!}C_{\phi,\nu} N^2<\infty,
	\end{equation*}
	from which we easily conclude $LF\in L^2(\mu)$.	
	
	We are left to prove that $U_t$ is differentiable on $\F$ and that its derivative at time $t=0$ is given by $L$.
	However, this is equivalent to show that $\omega_t=\TT_t\gamma$ solves \eqref{dssymm},
	for which we already referred to \cite{scho96}.
\end{proof}

Local observables $\F$ are not invariant for $\U_t$: this is due to the nonlinearity of the dynamics, not to
singularity of the interaction. Our techniques thus does not seem to be suited to this setting.

We conclude by mentioning an idea of \cite{scho96}, from which we quote:
\emph{``Considering point vortices to be solutions of the weak vorticity
	formulation allows us to extend their dynamics beyond collisions simply by 
	merging vortices that collide into a single vortex whose strength is the algebraic sum
	of the colliding vortices. Clearly this defines a solution for times less
	than and for times greater than the collision time, and the resulting vorticity
	is continuous in time in the weak-* topology of measures, so that there is no
	contribution [...] from the “jump” at the collision time. Of
	course, this extended notion of point-vortex dynamics is horribly nonunique since
	the time-reversibility of the Euler equations implies that a single vortex can split
	equally well into several vortices at any time.''}
Non uniqueness for the weak formulation of Euler equation in the point vortices case
might be a clue that $(\L,\F)$ is not essentially self-adjoint or even $L^2(\mu)$ unique.
However, producing counterexamples with collisions or splitting of vortices is a difficult problem:
explicit examples of collisions rely on integrability properties of the Hamiltonian dynamics.
Whether $(\L,\F)$ is essentially self-adjoint thus remains an interesting open question.

\bibliographystyle{plain}

\begin{thebibliography}{10}
	
	\bibitem{ai78}
	Michael Aizenman.
	\newblock A sufficient condition for the avoidance of sets by measure
	preserving flows in {${\bf R}^{n}$}.
	\newblock {\em Duke Math. J.}, 45(4):809--813, 1978.
	
	\bibitem{AlBaFe08}
	S.~Albeverio, V.~Barbu, and B.~Ferrario.
	\newblock Uniqueness of the generators of the 2{D} {E}uler and
	{N}avier-{S}tokes flows.
	\newblock {\em Stochastic Process. Appl.}, 118(11):2071--2084, 2008.
	
	\bibitem{AlBaFeERR}
	S.~Albeverio, V.~Barbu, and B.~Ferrario.
	\newblock Erratum to ``{U}niqueness of the generators of 2{D} {E}uler and
	{S}tokes flows'' [{S}tochastic {P}rocess. {A}ppl. 118 (11) (2008) 2071--2084]
	[mr2462289].
	\newblock {\em Stochastic Process. Appl.}, 120(10):2102, 2010.
	
	\bibitem{AlFe02}
	S.~Albeverio and B.~Ferrario.
	\newblock Uniqueness results for the generators of the two-dimensional {E}uler
	and {N}avier-{S}tokes flows. {T}he case of {G}aussian invariant measures.
	\newblock {\em J. Funct. Anal.}, 193(1):77--93, 2002.
	
	\bibitem{af03}
	S.~Albeverio and B.~Ferrario.
	\newblock 2{D} vortex motion of an incompressible ideal fluid: the
	{K}oopman-von {N}eumann approach.
	\newblock {\em Infin. Dimens. Anal. Quantum Probab. Relat. Top.},
	6(2):155--165, 2003.
	
	\bibitem{AlKoRo98}
	S.~Albeverio, Yu.~G. Kondratiev, and M.~R\"{o}ckner.
	\newblock Analysis and geometry on configuration spaces.
	\newblock {\em J. Funct. Anal.}, 154(2):444--500, 1998.
	
	\bibitem{ac90}
	Sergio Albeverio and Ana~Bela Cruzeiro.
	\newblock Global flows with invariant ({G}ibbs) measures for {E}uler and
	{N}avier-{S}tokes two-dimensional fluids.
	\newblock {\em Comm. Math. Phys.}, 129(3):431--444, 1990.
	
	\bibitem{Ar07}
	Hassan Aref.
	\newblock Point vortex dynamics: a classical mathematics playground.
	\newblock {\em J. Math. Phys.}, 48(6):065401, 23, 2007.
	
	\bibitem{co01}
	Lawrence Conlon.
	\newblock {\em Differentiable manifolds}.
	\newblock Birkh\"{a}user Advanced Texts: Basler Lehrb\"{u}cher. [Birkh\"{a}user
	Advanced Texts: Basel Textbooks]. Birkh\"{a}user Boston, Inc., Boston, MA,
	second edition, 2001.
	
	\bibitem{del91}
	Jean-Marc Delort.
	\newblock Existence de nappes de tourbillon en dimension deux.
	\newblock {\em J. Amer. Math. Soc.}, 4(3):553--586, 1991.
	
	\bibitem{dp82}
	D.~D\"{u}rr and M.~Pulvirenti.
	\newblock On the vortex flow in bounded domains.
	\newblock {\em Comm. Math. Phys.}, 85(2):265--273, 1982.
	
	\bibitem{flandoli}
	Franco Flandoli.
	\newblock Weak vorticity formulation of 2{D} {E}uler equations with white noise
	initial condition.
	\newblock {\em Comm. Partial Differential Equations}, 43(7):1102--1149, 2018.
	
	\bibitem{fk19}
	Stefan Fleischer and Andreas Knauf.
	\newblock Improbability of {C}ollisions in n-{B}ody {S}ystems.
	\newblock {\em Arch. Ration. Mech. Anal.}, 234(3):1007--1039, 2019.
	
	\bibitem{gp76}
	G.~Gallavotti and M.~Pulvirenti.
	\newblock Classical {KMS} condition and {T}omita-{T}akesaki theory.
	\newblock {\em Comm. Math. Phys.}, 46(1):1--9, 1976.
	
	\bibitem{ggm80}
	K.~Goodrich, K.~Gustafson, and B.~Misra.
	\newblock On converse to {K}oopman's lemma.
	\newblock {\em Phys. A}, 102(2):379--388, 1980.
	
	\bibitem{Gr19}
	Francesco {Grotto}.
	\newblock {Stationary Solutions of Damped Stochastic 2-dimensional Euler's
		Equation}.
	\newblock {\em arXiv e-prints}, page arXiv:1901.06744, Jan 2019.
	
	\bibitem{GrRo19}
	Francesco {Grotto} and Marco {Romito}.
	\newblock {A Central Limit Theorem for Gibbsian Invariant Measures of 2D Euler
		Equation}.
	\newblock {\em arXiv e-prints}, page arXiv:1904.01871, Apr 2019.
	
	\bibitem{GuPe18}
	Massimiliano {Gubinelli} and Nicolas {Perkowski}.
	\newblock {The infinitesimal generator of the stochastic Burgers equation}.
	\newblock {\em arXiv e-prints}, page arXiv:1810.12014, Oct 2018.
	
	\bibitem{KiNe98}
	Rangachari Kidambi and Paul~K. Newton.
	\newblock Motion of three point vortices on a sphere.
	\newblock {\em Phys. D}, 116(1-2):143--175, 1998.
	
	\bibitem{ArKi19}
	Alexander~A. Kilin and Elizaveta~M. Artemova.
	\newblock Integrability and chaos in vortex lattice dynamics.
	\newblock {\em Regul. Chaotic Dyn.}, 24(1):101--113, 2019.
	
	\bibitem{lem12}
	Mariusz Lema\'{n}czyk.
	\newblock Spectral theory of dynamical systems.
	\newblock In {\em Mathematics of complexity and dynamical systems. {V}ols.
		1--3}, pages 1618--1638. Springer, New York, 2012.
	
	\bibitem{lionsbook}
	Pierre-Louis Lions.
	\newblock {\em On {E}uler equations and statistical physics}.
	\newblock Cattedra Galileiana. [Galileo Chair]. Scuola Normale Superiore,
	Classe di Scienze, Pisa, 1998.
	
	\bibitem{mpp78}
	C.~Marchioro, A.~Pellegrinotti, and M.~Pulvirenti.
	\newblock Selfadjointness of the {L}iouville operator for infinite classical
	systems.
	\newblock {\em Comm. Math. Phys.}, 58(2):113--129, 1978.
	
	\bibitem{mp94}
	Carlo Marchioro and Mario Pulvirenti.
	\newblock {\em Mathematical theory of incompressible nonviscous fluids},
	volume~96 of {\em Applied Mathematical Sciences}.
	\newblock Springer-Verlag, New York, 1994.
	
	\bibitem{PoDr93}
	Lorenzo~M. Polvani and David~G. Dritschel.
	\newblock Wave and vortex dynamics on the surface of a sphere.
	\newblock {\em J. Fluid Mech.}, 255:35--64, 1993.
	
	\bibitem{rs75i}
	Michael Reed and Barry Simon.
	\newblock {\em Methods of modern mathematical physics. {I}. {F}unctional
		analysis}.
	\newblock Academic Press, New York-London, 1972.
	
	\bibitem{sa73}
	Donald~G. Saari.
	\newblock Improbability of collisions in {N}ewtonian gravitational systems.
	{II}.
	\newblock {\em Trans. Amer. Math. Soc.}, 181:351--368, 1973.
	
	\bibitem{sa71}
	Donald~Gene Saari.
	\newblock Improbability of collisions in {N}ewtonian gravitational systems.
	\newblock {\em Trans. Amer. Math. Soc.}, 162:267--271; erratum, ibid. 168
	(1972), 521, 1971.
	
	\bibitem{sch95}
	Steven Schochet.
	\newblock The weak vorticity formulation of the {$2$}-{D} {E}uler equations and
	concentration-cancellation.
	\newblock {\em Comm. Partial Differential Equations}, 20(5-6):1077--1104, 1995.
	
	\bibitem{scho96}
	Steven Schochet.
	\newblock The point-vortex method for periodic weak solutions of the 2-{D}
	{E}uler equations.
	\newblock {\em Comm. Pure Appl. Math.}, 49(9):911--965, 1996.
	
	\bibitem{ArSt99}
	Mark~A. Stremler and Hassan Aref.
	\newblock Motion of three point vortices in a periodic parallelogram.
	\newblock {\em J. Fluid Mech.}, 392:101--128, 1999.
	
	\bibitem{tel17}
	A.~F.~M. ter Elst and M.~Lema\'{n}czyk.
	\newblock On one-parameter {K}oopman groups.
	\newblock {\em Ergodic Theory Dynam. Systems}, 37(5):1635--1656, 2017.
	
\end{thebibliography}

\end{document}